%% file: 4.tex
\documentclass[twocolumn, journal]{IEEEtran}
 
\def\baselinestretch{1}  

\input{Preamble/preamble}

\begin{document}

\title{Formulating Connectedness in Security-Constrained Optimal Transmission Switching Problems}
\author{Tong Han,
        David J. Hill,~\IEEEmembership{Life Fellow, IEEE},
        and Yue Song,~\IEEEmembership{Member, IEEE}
 \thanks{This work was supported by the Research Grants Council of the Hong Kong Special Administrative Region through the General Research Fund under Project No. 17209419.}
 \thanks{T. Han, D. J. Hill, and Y. Song are with the Department of Electrical and Electronic Engineering, University of Hong Kong, Hong Kong (e-mail: hantong@eee.hku.hk; yuesong@eee.hku.hk; dhill@eee.hku.hk).}
 \thanks{D. J. Hill is also with the School of Electrical Engineering and Telecommunications, The University of New South Wales, Kensington, NSW 2052, Australia (e-mail: dhill@eee.hku.hk).} 
}

\maketitle

\begin{abstract}
    This paper focuses on the issue of network connectedness (NC) in security-constrained optimal transmission switching problems, which is complicated by branch contingencies and corrective line switching. Two criteria are firstly proposed with the principle of preserving NC as much as possible within reasonable limits. By extending the electrical flow based NC constraints, a proposition is derived to associate different cases of NC with the optimum of a linear program, yielding the mathematical formulation of the NC criteria. By Karush-Kuhn-Tucker conditions,  this formulation is further transformed into a tractable version which can be incorporated with existing SCOTS models without affecting the applicability of original solution approaches. Finally, case studies on various networks and SCOTS models demonstrate the efficacy of the proposed approach.
\end{abstract}


\IEEEpeerreviewmaketitle

\section*{Nomenclature}

\begin{IEEEdescription}[\IEEEusemathlabelsep\IEEEsetlabelwidth{$V_1V$}]   
\item[$\!\!\!\!\!\!$$\mathcal{G}(\mathcal{V}, \mathcal{E})$] The undirected graph representing the transmission network topologically where $\mathcal{V}$ and $\mathcal{E}$ are sets of buses and branches, respectively. 
\item[$\!\!\!\!\!\!$$n_{\rm s}$] Number of connected node-induced subgraphs in $\mathcal{G}$.
\item[$\!\!\!\!\!\!$$M$] Big-M constant.  
\item[$\!\!\!\!\!\!$$n_{\rm n},\! n_{\rm g},\! n_{\rm b}$] Numbers of buses, generators, and branches.
\item[$\!\!\!\!\!\!$$\eta$, $\lambda$] Maximal number of fault components, $\lambda \in \mathbb{Z}^{+}$ 
\item[$\!\!\!\!\!\!$$\bm{p}_{\rm g}$] Active power outputs of generators. 
\item[$\!\!\!\!\!\!$$\bm{z}$] Statuses of branches. An entry value of 1/0 represents the associated branch is switched on/off.
\item[$\!\!\!\!\!\!$$\tilde{\bm{z}}, \bar{\bm{z}}$]  Counterpart of $\bm{z}$ for the post-contingency/control topology.
\item[$\!\!\!\!\!\!$$\bm{p}_{\rm g_+}$, $\bm{p}_{\rm g_-}$] Upward/downward regulations of active power outputs of generators.   
\item[$\!\!\!\!\!\!$$\bm{z}_{+}$, $\bm{z}_{-}$]  Action signs of switching on/off branches. An entry value of 1/0 means a/no switching action performed.
\item[$\!\!\!\!\!\!$$\bm{o}$]  Parameterization of $N\!-\!k$ contingencies. Entry values of 1/0 indicate the normal/failure state of components.
\item[$\!\!\!\!\!\!$$\bm{o}_{\rm g}$, $\bm{o}_{\rm b}$]  Sub-vectors of $\bm{o}$ for generators and branches. 
\item[$\!\!\!\!\!\!$$\bm{E}$] Oriented incidence matrix of graph $\mathcal{G}$ with each branch assigned arbitrary and fixed orientation.
\item[$\!\!\!\!\!\!$$f(\cdot), g(\cdot)$] Function of dispatch cost and corrective control cost.
\item[$\!\!\!\!\!\!$$P$, $\mathcal{P}$] Probability distribution of $\bm{o}$ and its ambiguity set.
\item[$\!\!\!\!\!\!$$\epsilon$, $\bm{1}$, $\bm{H}$] $\epsilon >0$, all-ones vector, all-ones matrix. 
\end{IEEEdescription}
{
\footnotesize{
\textit{Note}: Except for $\bm{z}$, $\tilde{\bm{z}}$, $\bar{\bm{z}}$, $\bm{z}_+$, $\bm{z}_-$ and $\bm{o}_{\rm b} \in \mathbb{B}^{n_{\rm b}}$, $\bm{o}_{\rm g} \in \mathbb{B}^{n_{\rm g}} $, and $\bm{o} \in \mathbb{B}^{n_{\rm g} + n_{\rm b}} $, other bold lowercase letters are vectors in $\mathbb{R}$ with proper dimension.}}

\section{Introduction}

\IEEEPARstart{G}{rowing} penetration of variable renewable energy decreases generation-side dispatchablility. Optimal transmission switching (OTS), which leverages grid-side flexibility to improve economic performance, is going to be engaged in future network operations more widely and actively.

The problem of how to formulate network connectedness (NC) constraints in a strict and tractable way is essential for OTS problems. Two recent works filled this gap by approaches based on electrical flow and Miller-Tucker-Zemlin constraints \cite{4-995-ea, 4-1283}, respectively. 
However, for more practical security-constrained OTS (SCOTS) problems, the issue of NC becomes more complicated due to branch contingencies and potential corrective line switching. For instance, for $N-k$ contingencies with a relatively large $k$, it is infeasible to always ensure NC. Additionally, due to inherent local weak connection, even an $N-1$ contingency on the network with all branches switched on can cause network disconnection (ND). These problems hinder ensuring NC for all topologies after contingencies and corrective line switching and thus a direct application of NC constraints in \cite{4-995-ea} and \cite{4-1283}. Therefore, SCOTS problems require new NC criteria and their tractable mathematical formulation which can be incorporated into existing SCOTS models. 

This paper addresses the NC issue in SCOTS problems. First, two criteria for NC of topologies in SCOTS are proposed considering the actual situation in transmission networks. Then, these two criteria are mathematically formulated by extending the electrical flow based NC constraints \cite{4-995-ea}. Specifically, by designing parameters in the electrical flow based NC constraints and constructing a linear program (LP), we derive a proposition to associate different cases of NC with the optimum of the LP. This proposition yields a tractable mathematical formulation of the criteria. Finally, we demonstrate the efficacy of the proposed approach numerically.

\section{Models and Connectedness Problems of SCOTS}

\subsection{Models of SCOTS}

We first introduce the SCOTS model to make the paper self-contained. Depending on the way contingencies are treated, SCOTS can be a stochastic, robust, or distributionally robust (DR) optimization problem. More generally, we adopt the two-stage DR formulation of SCOTS, 
written as 
\begin{subequations}\label{eq-0-2-1} 
    \vspace{-6pt} 
    \setlength{\belowdisplayskip}{1pt} 
    \begin{align}
        \!\!\!\! \min_{\bm{p}_{\rm g}, \bm{z}} ~&  f(\bm{p}_{\rm g} ) + \sup_{P \in \mathcal{P}} \!\mathbb{E}_{\!P}  
        \left[ \begin{aligned}
            \min_{\bm{x} } ~&  g( \bm{x} )~\text{s.t.}~  \bm{x} \!\in\! \mathcal{Z}(\bm{p}_{\rm g},\! \bm{z},\! \bm{o})  
        \end{aligned}   \!\!\!\! \right] \\[-2pt]
        \!\!\!\! \text{s.t.} ~
        &  (\bm{p}_{\rm g}, \bm{z}) \in \mathcal{X}  
    \end{align} 
\end{subequations}
where $\bm{x} \!\!=\!\! [ \bm{p}_{\rm g_+}^T, \bm{p}_{\rm g_-}^T, \bm{p}_{\rm d_{\!\Delta}}^T, \bm{z}_{+}^T,  \bm{z}_{-}^T  ]^T$, and the first and second stages are the dispatch problem under the normal state and corrective control problem after a contingency, respectively. When $\mathcal{P}$ contains only the true distribution of $\bm{o}$, (\ref{eq-0-2-1}) reduces to the stochastic version of SCOTS. On the other hand, if $\mathcal{P}$ contains all probability distributions on the support of $\bm{o}$, then (\ref{eq-0-2-1}) reduces to the robust version of SCOTS.

\subsection{Network Connectedness Problems in SCOTS}
For SCOTS, the issue of NC arises from three operating states, i.e., the normal, post-contingency, and post-control states. Fig. \ref{fig-0-2-1} illustrates topologies under these three states. Hereafter, we also use $\bm{z}$, $\tilde{\bm{z}}$ and $\bar{\bm{z}}$ to refer to the corresponding topologies for brevity. NC constraints proposed in \cite{4-995-ea} can be used to ensure NC of $\bm{z}$, while for $\tilde{\bm{z}}$ and $\bar{\bm{z}}$, these constraints are inapplicable due to the following practical situations: 

\textit{(1)} For many transmission networks, even an $N-1$ contingency on the network with all branches switched on can cause ND. This generally creates a main connected component containing almost all buses, and other components with a few buses. It is reasonable to ignore such \textit{inevitable ND} (see Fig. \ref{fig-0-2-1} for explanation of this concept) since it is caused by inherent local weak connection instead of line switching. For stochastic SCOTS generally formulated with pre-defined contingency scenarios, by identifying contingencies causing inevitable ND beforehand, NC constraints in \cite{4-995-ea} may still be applicable. However, for the robust and DR versions which are generally not scenario-based, such pre-identification is infeasible.

\textit{(2)} It is generally impossible to ensure NC of $\tilde{\bm{z}}$ for all $N-k$ contingencies when the maximal number of fault components is relatively large. Reasonably, NC can be ensured only when the number of fault components is below a threshold. In addition, NC of $\bar{\bm{z}}$ also cannot be always ensured since the post-contingency topology may be disconnected. Thus requirements for NC of $\bar{\bm{z}}$ should depend on NC of the post-contingency topology.

\begin{figure}[t]
	\centering
	\includegraphics[width=8.8cm]{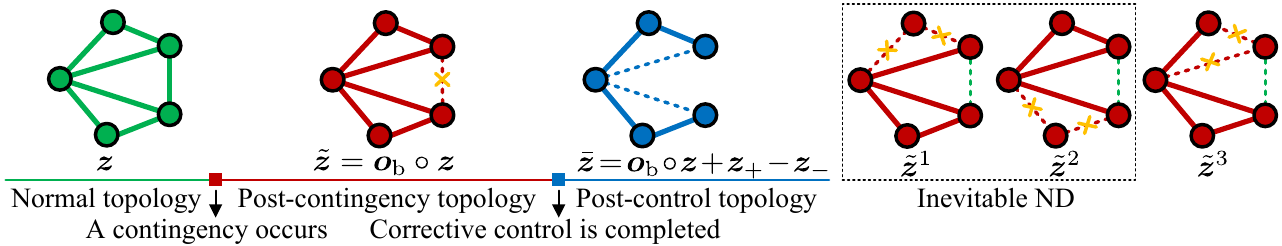} 
	\caption{Illustration of topologies under normal, post-contingency, and post-control states, and inevitable ND. Assume that the green graph is the network with all branches switched on and the dashed green line is switched off by OTS. Three different contingencies result in disconnected topologies $\tilde{\bm{z}}^1$, $\tilde{\bm{z}}^2$ and $\tilde{\bm{z}}^3$. We say that $\tilde{\bm{z}}^1$ and $\tilde{\bm{z}}^2$ are \textit{inevitably disconnected} since even the green line is switched on, these NDs will occurs. We call such ND \textit{inevitable ND}. In contrast, $\tilde{\bm{z}}^3$ is not inevitably disconnected since if the green line is switched on, $\tilde{\bm{z}}^3$ will be connected.}
	\label{fig-0-2-1} 
\end{figure}

\section{Criteria and Formulation of Connectedness}

\subsection{Criteria for Network Connectedness in SCOTS}

Define \textit{$\lambda$-branch contingencies} as the contingencies whose number of fault branches is not more than $\lambda$. 
Then, considering points \textit{(1)} and \textit{(2)}, we propose two criteria for NC in SCOTS problems as follows: 
\begin{criterion}\label{crit-0-2-1}
    $\tilde{\bm{z}}$ is connected for all $\lambda$-branch contingencies, ignoring inevitable NDs.
\end{criterion} 
\begin{criterion}\label{crit-0-2-2}
    Corrective line switching should not further disconnect the network when $\tilde{\bm{z}}$ is connected or disconnected with only inevitable ND.
\end{criterion} 
\begin{remark}
    Criterion \ref{crit-0-2-2} indicates that when $\tilde{\bm{z}}$ is connected, $\bar{\bm{z}}$ should be connected; when $\tilde{\bm{z}}$ is disconnected with only inevitable ND, $\bar{\bm{z}}$ should not create new ND; and otherwise, $\bar{\bm{z}}$ can be connected or not. The principle of the above criteria is to preserve NC after contingencies and corrective line switching as much as possible within reasonable limits.
\end{remark}

\subsection{Formulation of Network Connectedness in SCOTS}

The two criteria above are further formulated mathematically, based on the NC constraints proposed in \cite{4-995-ea}. To begin with, we introduce a parameterized region $\mathcal{C}(\phi, \bm{c}, \bm{d}) = $
\begin{equation}\label{eq-0-2-4} 
    \!\!\! \left\{\!\! 
    \bm{u} \!\!\in\!\! \mathbb{B}^{n_{\rm b}} \!\! \left|
   \begin{aligned} 
       & M (\bm{u}  \!-\! \bm{1} \!-\!  \phi \bm{1}) \!\!\leq\!\! \bm{E}^T \!\bm{\vartheta} \!-\! \bm{\rho}  \!\!\leq\!\! M (\bm{1} \!-\! \bm{u} \!+\!   \phi \bm{1})  \\  
       & \!-\! M (\bm{u} \!+\! \phi \bm{1}) \!\leq\! \bm{\rho} \!\leq\! M (\bm{u} \!+\! \phi \bm{1})  \\
       & \!-\! \phi M \bm{1} \!\!\leq\!\! \bm{E} \bm{\rho} \!-\! (\bm{c} \!+\! \bm{d}) \!\!\leq\!\! \phi M \bm{1}, \bm{\vartheta} \!\in\! \mathbb{R}^{n_{\rm n}}\!, \bm{\rho} \!\in\! \mathbb{R}^{n_{\rm b}}
   \end{aligned} \!\!\!
   \right.
   \right\}
   \vspace{-2pt}
\end{equation}
with $\phi \!\!\in\!\! \mathbb{R}$, $\bm{d} \!\!\in\!\! \mathbb{R}^{n_{\rm n}}$, and $\bm{c}$ being an $n_{\rm n}$-dimensional constant uniquely-balanced vector (see \cite[Definition~1]{4-995-ea} for its definition). When $\phi\!=\!0$ and $\bm{d} \!=\! \bm{0}$, constraints in (\ref{eq-0-2-4}) reduce to the NC constraints in \cite{4-995-ea}, such that $\mathcal{C}$ is the region of $\bm{u}$ whose associated topology is connected. When $\phi\geq1$ and $\bm{d} = \bm{0}$, constraints in (\ref{eq-0-2-4}) are all invalid. The general idea below to formulate the two criteria is to design $\bm{c}$ particularly rather than just being uniquely-balanced, such that $\bm{d}$, satisfying certain conditions, can be used to identify and ensure more complex cases of NC.

Let $\mathcal{W}(\lambda) \!=\! \{ (\mathcal{L}_i, \mathcal{N}_i)| i\!=\!1,2,...,n_{\rm w} \}$ be the set of all $n_{\rm w}$ pairs of ($\mathcal{L}_i, \mathcal{N}_i$) that satisfy the following three conditions:
\textit{(\romannumeral1)} $\mathcal{L}_i \!\subset\! \mathcal{E}$, $\mathcal{N}_i \!\subset\! \mathcal{V}$, and $|\mathcal{L}_i| \!\leq\! \lambda$; 
\textit{(\romannumeral2)} removing ${\mathcal{L}}_i$ from $\mathcal{G}(\mathcal{V}, \mathcal{E})$ causes ND with $\mathcal{N}_i$ being the set of buses not in the main connected component;
\textit{(\romannumeral3)} $\nexists \mathcal{L}' \subset \mathcal{L}_i$ and $\mathcal{L}' \neq \emptyset$, removing $\mathcal{L}_i \backslash \mathcal{L}'$ from $\mathcal{G}(\mathcal{V}, \mathcal{E})$ causes ND with $\mathcal{N}_i$ being the set of buses not in the main connected component. 
For short, NDs in condition \textit{(\romannumeral2)} are referred to as $\mathcal{W}(\lambda)$ \textit{NDs} 
and the associated graphs are $\mathcal{W}(\lambda)$-\textit{disconnected}. 
Thereby, the set of $\mathcal{W}(\lambda)$ NDs and that of all inevitable NDs in Criterion \ref{crit-0-2-1} are identical. Note that lines in $\cup_{i=1}^{n_{\rm w}} \mathcal{L}_i$ are assumed not to be switched off in $\bm{z}$ since such switching facilitates inevitable ND.

Next, introduce matrix $\bm{J}$ being the same as $\bm{J} \in \mathbb{R}^{n_{\rm s} \times n_{\rm n}}$ in \cite{4-995-ea}, $\bm{E}_{\rm w} \in \mathbb{R}^{n_{\rm w} \times n_{\rm n} }$ satisfying that $\forall i = 1,2,...,n_{\rm w}$, $\bm{E}_{{\rm w}, ij} = 1$ if $j \in \mathcal{N}_i$ and $\bm{E}_{{\rm w}, ij} = 0$ otherwise, and $\bm{J}_{\rm w}$ formed by deleting the common rows between $\bm{J}$ and $\bm{E}_{\rm w}$ and between $\bm{J}$ and $\bm{H} - \bm{E}_{\rm w}$, and the all-ones row, from $\bm{J}$. Denote by $n_{\rm u}$ the maximal number of connected components among all $\mathcal{W}(\lambda)$ NDs, and $n_{\rm d}$ the row dimension of $\bm{J}_{\rm w}$. Then a balance property of $\bm{c}$ is defined by Definition \ref{def-0-2-1}. 

\begin{definition}[$\mathcal{W}(\lambda)$-unique balance]\label{def-0-2-1}
    Let $\bm{b} = \bm{J} \bm{c} \in \mathbb{R}^{n_{\rm s}}$. Then $\bm{c}$ is $\mathcal{W}(\lambda)$-uniquely balanced if $\bm{b}_1=0$ and $\bm{b}_i\neq0$ for $i= 2,3,..., n_{\rm s}$, namely that $\bm{c}$ is uniquely balanced; and $\exists r \in \mathbb{R}$, \text{s.t.}, $\Vert \bm{E}_{\rm w} \bm{c} \Vert_{\infty} \leq r$ and $\Vert \bm{J}_{\rm w} \bm{c} \Vert_{-\infty} \geq n_{\rm u} r$. 
\end{definition}

\begin{remark}
    Unlike uniquely balanced $\bm{c}$ which has a trivial special case \cite{4-995-ea}, it may be far from easy to find such a special case for $\mathcal{W}(\lambda)$-uniquely balanced $\bm{c}$. An ad hoc approach is to solve the following mixed-integer LP:
    \begin{subequations}\label{eq-0-2-c-milp} 
        \begin{align}
            & \min_{\bm{c} \in \mathbb{R}^{n_{\rm n}}, \bm{b} \in \mathbb{R}^{n_{\rm s}}, \bm{\beta} \in \mathbb{B}^{n_{\rm s} - 1}, \bm{\gamma} \in \mathbb{B}^{n_{\rm d}}, r \in \mathbb{R} } ~~ r  \\
            \rm{s.t.} ~& \bm{b} = \bm{J} \bm{c}, \bm{b}_1 = 0, -r \bm{1}  \leq \bm{E}_{\rm w} \bm{c}  \leq r \bm{1} \\ 
            & \epsilon \!-\!  M\! \bm{\beta}_{i-1} \!\leq\! \bm{b}_i  \!\leq\! -  \epsilon \!+\! M (1 \!-\! \bm{\beta}_{i-1}), i=2, 3,..., n_{\rm s} \\
            &  n_{\rm u} r \bm{1} -  M \bm{\gamma} \leq  \bm{J}_{\rm w} \bm{c}  \leq - n_{\rm u} r \bm{1}  + M (1 -\bm{\gamma}) 
        \end{align}
    \end{subequations}
    where $\bm{\beta}$, $\bm{\gamma}$ and $r$ are auxiliary variables. This LP is derived from Definition \ref{def-0-2-1}, with the objective function set as any one with a finite lower bound, and constraints in the form $|\cdot| > 0$ being linearized. Any feasible solution of (\ref{eq-0-2-c-milp}) yields a $\mathcal{W}(\lambda)$-uniquely balanced $\bm{c}$.
\end{remark}

With $\bm{c}$ being an $n_{\rm n}$-dimensional constant $\mathcal{W}(\lambda)$-uniquely balanced vector, consider the following $\tilde{\bm{z}}$-parameterized LP: 
\begin{equation}\label{eq-0-2-5} 
        \!\! \min_{\bm{d}_+, \bm{d}_- \!\in \mathbb{R}^{n_{\rm n}}} \!\!\! \bm{1}^{\!T} \!(\!\bm{d}_+ \!\!+ \bm{d}_- \!)  ~\text{s.t.} ~ \mathcal{C}(0, \bm{c}, \bm{d}_+ \!\!-\! \bm{d}_-) \!\!\ni\!\! \tilde{\bm{z}}, \bm{d}_+ \!\!\geq\!\! \bm{0}, \bm{d}_- \!\!\geq\!\! \bm{0}  
\end{equation}
Denote by $\mathcal{D}(\tilde{\bm{z}})$ the set of all optima of (\ref{eq-0-2-5}), and $\Vert \bm{b} \Vert_{-2 \infty}$ the second smallest absolute value of entries of $\bm{b}$. Then NC of $\tilde{\bm{z}}$ and the optima of (\ref{eq-0-2-5}) are associated by Proposition \ref{prop-0-2-1}.

\begin{proposition}\label{prop-0-2-1}
    $\forall (\bm{d}_+^*, \bm{d}_-^*) \in \mathcal{D}(\tilde{\bm{z}})$, (\uppercase\expandafter{\romannumeral1}) $\bm{1}^T (\bm{d}_+^* + \bm{d}_-^*) = 0$ iff $\tilde{\bm{z}}$ is connected; 
    (\uppercase\expandafter{\romannumeral2}) $\bm{1}^T (\bm{d}_+^* + \bm{d}_-^*) \geq 2 \Vert \bm{b} \Vert_{-2 \infty} > 0$ iff $\tilde{\bm{z}}$ is disconnected; 
    and (\uppercase\expandafter{\romannumeral3}) $\bm{1}^T (\bm{d}_+^* + \bm{d}_-^*) \leq n_{\rm u} r$ iff $\tilde{\bm{z}}$ is connected or $\mathcal{W}(\lambda)$-disconnected. 
\end{proposition}

\begin{proof}  
    \textit{(\uppercase\expandafter{\romannumeral1})}. (Sufficiency) By $\bm{d}_+^* \geq 0$, $\bm{d}_-^* \geq 0$, and $\bm{1}^T (\bm{d}_+^* + \bm{d}_-^*) = 0$, we have $\bm{d}_+^* = \bm{d}_-^* = \bm{0}$, and thus $\tilde{\bm{z}} \in \mathcal{C}(0, \bm{c}, \bm{0})$, i.e., $\tilde{\bm{z}}$ is connected. (Necessity) Ignoring constraint $\tilde{\bm{z}} \in \mathcal{C}(\cdot)$ in (\ref{eq-0-2-5}), the unique optimum of (\ref{eq-0-2-5}) is $(\bm{d}_+^*, \bm{d}_-^*) = (\bm{0}, \bm{0})$. Since program (\ref{eq-0-2-5}) with $\tilde{\bm{z}}$ being connected is also feasible at $(\bm{d}_+^*, \bm{d}_-^*)$, $(\bm{d}_+^*, \bm{d}_-^*)$ is also the unique optimum of this program. Therefore, if $\tilde{\bm{z}}$ is connected, $\bm{1}^T (\bm{d}_+^* + \bm{d}_-^*) = 0$.

    \textit{(\uppercase\expandafter{\romannumeral2})}. (Sufficiency) By \textit{(\uppercase\expandafter{\romannumeral1})}, $\bm{1}^T (\bm{d}_+^* + \bm{d}_-^*) \neq 0$ yields that $\tilde{\bm{z}}$ is disconnected. (Necessity) By the proof of Theorem 1 in \cite{4-995-ea}, constraints in $\tilde{\bm{z}} \in \mathcal{C}(0, \bm{c}, \bm{d}_+ - \bm{d}_-)$ are equivalent to 
    \begin{equation}\label{eq-0-2-6}
        \bm{L}_{\mathcal{G}_{\tilde{\bm{z}}}} \bm{\vartheta} = \bm{c} + \bm{d}_+ - \bm{d}_-
    \end{equation}
    where $\bm{L}_{\mathcal{G}_{\tilde{\bm{z}}}}$ is the Laplacian matrix of graph $\mathcal{G}(\mathcal{V}, \tilde{\mathcal{E}})$ with $\tilde{\mathcal{E}}$ being the set of branches in $\tilde{\bm{z}}$. Assume that graph $\mathcal{G}(\mathcal{V}, \tilde{\mathcal{E}})$ contains $n_{\rm c}$ connected components, and let $[\bm{L}_{\mathcal{G}_{\tilde{\bm{z}}}}]_i$ be the principle submatrix of $\bm{L}_{\mathcal{G}_{\tilde{\bm{z}}}}$ associated with the $i$-th connected component, $[\bm{\vartheta}]_i$ be the subvector of $\bm{\vartheta}$ associated with the $i$-th connected component, and $[\bm{c}]_i$, $[\bm{d}_+]_i$, and $[\bm{d}_-]_i$ are analogous. Then (\ref{eq-0-2-6}) can be reformulated as
    \begin{equation}\label{eq-0-2-7}
        \!\!\!\!\!\! \begin{bmatrix}
            [\bm{L}_{\mathcal{G}_{\tilde{\bm{z}}}}]_1 \!\!\!\!&\!\!\!\!\cdots\!\!\!\!&\!\!\!\!\cdots\!\!&\!\! 0 \\
            \vdots\!\!\!\!&\!\!\!\! [\bm{L}_{\mathcal{G}_{\tilde{\bm{z}}}}]_2 \!\!\!\!&\!\!\!\!\!\!&\!\! \vdots  \\
            \vdots \!\!\!\!&\!\!\!\!\!\!\!\!&\!\!\!\! \ddots \!\!&\!\! \vdots \\
            0 \!\!\!\!&\!\!\!\!\cdots\!\!\!\!&\!\!\!\!\cdots\!\!&\!\! [\bm{L}_{\mathcal{G}_{\tilde{\bm{z}}}}]_{n_{\rm c}}
        \end{bmatrix} \!\!\!
        \begin{bmatrix}
            [\bm{\vartheta}]_{1} \\
            [\bm{\vartheta}]_{2} \\
            \vdots\\
            [\bm{\vartheta}]_{\!n_{\rm c}}  \!
        \end{bmatrix}
        \!\!\!=\!\!\! 
        \begin{bmatrix}
            [\bm{c}]_1 \!\!+\! [\bm{d}_+]_1 \!\!-\! [\bm{d}_-]_1 \\
            [\bm{c}]_2 \!\!+\! [\bm{d}_+]_2 \!\!-\! [\bm{d}_-]_2 \\
            \vdots \\
            [\bm{c}]_{\!n_{\rm c}} \!\!\!+\!\! [\bm{d}_+]_{\!n_{\rm c}} \!\!\!-\!\! [\bm{d}_-]_{\!n_{\rm c}} 
        \end{bmatrix}
        \!\!\!
    \end{equation}
    which together with the definition of Laplacian matrices, gives
    \begin{equation}\label{eq-0-2-8}
        \bm{1}^T ( [\bm{c}]_i + [\bm{d}_+]_i - [\bm{d}_-]_i ) = 0, i=1, 2, ..., n_{\rm c}
    \end{equation}
    Furthermore, for any connected component of $\mathcal{G}(\mathcal{V}, \tilde{\mathcal{E}})$, there exists a node induced subgraph of $\mathcal{G}(\mathcal{V}, \mathcal{E})$ whose node set is equal to that of the connected component and unequal to $\mathcal{V}$. By the fact that $\bm{c}$ is uniquely balanced, we have $\forall i=1,2,...,n_{\rm c}$, $\bm{1}^T [\bm{c}]_i = \bm{b}_{j(i)}$ with $j(i) \in \{2, 3,..., n_{\rm s}\}$, and thus 
    \begin{equation}\label{eq-0-2-9}
        \bm{1}^T ([\bm{d}_+]_i - [\bm{d}_-]_i ) = \bm{b}_{j(i)} , i=1, 2, ..., n_{\rm c}
    \end{equation}
    Then (\ref{eq-0-2-5}) is equivalent to 
    \begin{subequations}\label{eq-0-2-10}
        \begin{align}
            \!\! \min_{\bm{d}_+ \in \mathbb{R}^{n_{\rm n}}, \bm{d}_- \in \mathbb{R}^{n_{\rm n}}} & \bm{1}^T (\bm{d}_+ + \bm{d}_- ) \label{eq-0-2-10:1} \\
            \text{s.t.} ~& \bm{1}^T \! ([\bm{d}_+]_i \!-\! [\bm{d}_-]_i ) \!=\! \bm{b}_{j(i)} , i\!=\!1, 2, ..., n_{\rm c} \label{eq-0-2-10:2} \\
            & \bm{d}_+ \geq \bm{0}, \bm{d}_- \geq \bm{0} \label{eq-0-2-10:3}
        \end{align}
    \end{subequations}
    To obtain the optimum of (\ref{eq-0-2-10}), we write its dual problem as
    \begin{subequations}\label{eq-0-2-11}
        \begin{align}
            \max_{\bm{\zeta} \in \mathbb{R}^{n_{\rm c}}} & \sum_{i=1}^{n_{\rm c}} \bm{\zeta}_i \bm{b}_{j(i)} ~~\text{s.t.}~ -1 \leq \bm{\zeta}_i \leq  1, i=1,2,..., n_{\rm c}
        \end{align} 
    \end{subequations}
    Consider a solution of (\ref{eq-0-2-10}), denoted as $(\bm{d}_+^{\star}, \bm{d}_-^{\star})$, which satisfies that  $\forall i=1,2,..., n_{\rm c}$, if $\bm{b}_{j(i)} > 0$, then $[\bm{d}_-^{\star}]_{i} = \bm{0}$ and $[\bm{d}_+^{\star}]_{i}$ contains only one nonzero element being $\bm{b}_{j(i)}$; and if $\bm{b}_{j(i)} < 0$, then $[\bm{d}_+^{\star}]_{i} = \bm{0}$ and $[\bm{d}_-^{\star}]_{i}$ contains only one nonzero element being $-\bm{b}_{j(i)}$. Also consider a solution of (\ref{eq-0-2-11}), denoted as $\bm{\zeta}^{\star}$, which satisfies that $\forall i=1,2,..., n_{\rm c}$, if $\bm{b}_{j(i)} > 0$, then $\bm{\zeta}_i^{\star} = 1$, and if $\bm{b}_{j(i)} < 0$, then $\bm{\zeta}_i^{\star} = -1$. Since $(\bm{d}_+^{\star}, \bm{d}_-^{\star})$ and $\bm{\zeta}^{\star}$ are feasible solutions and $\bm{1}^T (\bm{d}_+^{\star} + \bm{d}_-^{\star} ) = \sum_{i=1}^{n_{\rm c}} \bm{\zeta}_i^{\star} \bm{b}_{j(i)}$, by strong duality, $(\bm{d}_+^{\star}, \bm{d}_-^{\star})$ is a global optimum of (\ref{eq-0-2-10}). Therefore, $\bm{1}^T (\bm{d}_+^* + \bm{d}_-^*) = \bm{1}^T (\bm{d}_+^{\star} + \bm{d}_-^{\star}) = \sum_{i=1}^{n_{\rm c}} |\bm{b}_{j(i)}| $, which together with $n_{\rm c} \geq 2$ and $j(i) \neq 1$, gives $\bm{1}^T (\bm{d}_+^* + \bm{d}_-^*) \geq 2 \Vert \bm{b} \Vert_{-2 \infty} > 0$.

    \textit{(\uppercase\expandafter{\romannumeral3})} (Sufficiency) By \textit{(\uppercase\expandafter{\romannumeral1})}, if $\bm{1}^{T} (\bm{d}_+^* + \bm{d}_-^*) = 0$, $\tilde{\bm{z}}$ is connected, and if $0 < \bm{1}^T (\bm{d}_+^* + \bm{d}_-^*) \leq n_{\rm u} r$, $\tilde{\bm{z}}$ is disconnected. Next, we prove that in the latter case, $\tilde{\bm{z}}$ is also $\mathcal{W}(\lambda)$ disconnected by contradiction. Assume that $\tilde{\bm{z}}$ is not $\mathcal{W}(\lambda)$ disconnected, then there must be at least one connected component associated with a row of $\bm{J}$ which is also in $\bm{J}_{\rm w}$. Following the proof of necessity of \textit{(\uppercase\expandafter{\romannumeral2})}, it indicates that $\exists i \in \{1,2,...,n_{\rm c}\}$ s.t. $|\bm{b}_{j(i)}| \geq \Vert \bm{J}_{\rm w} \bm{c} \Vert_{- \infty}$. Further by $\Vert \bm{J}_{\rm w} \bm{c} \Vert_{- \infty} \geq n_{\rm u} r$ , $|\bm{b}_{j(i)}| > 0$ with $i \in \{1,2,...,n_{\rm c}\}$, and $n_{\rm c} \geq 2$, we have $\bm{1}^T (\bm{d}_+^* + \bm{d}_-^*) = \sum_{i=1}^{n_{\rm c}} |\bm{b}_{j(i)}| > n_{\rm u} r$, which contradicts $\bm{1}^T (\bm{d}_+^* + \bm{d}_-^*) \leq n_{\rm u} r$. This proves that $\tilde{\bm{z}}$ is $\mathcal{W}(\lambda)$ disconnected, completing the proof of sufficiency.

    (Necessity) According to \textit{(\uppercase\expandafter{\romannumeral1})}, when $\tilde{\bm{z}}$ is connected, $\bm{1}^T (\bm{d}_+^* + \bm{d}_-^*) = 0 < n_{\rm u} r$. When $\tilde{\bm{z}}$ is $\mathcal{W}(\lambda)$ disconnected, following the proof of necessity of \textit{(\uppercase\expandafter{\romannumeral2})}, $\bm{1}^T (\bm{d}_+^* + \bm{d}_-^*) = \sum_{i=1}^{n_{\rm c}} |\bm{b}_{j(i)}|$. $\forall i \in \{1,2,..., n_{\rm c}\}$, the row of $\bm{J}$ associated with $\bm{b}_{j(i)}$ is in $\bm{E}_{\rm w}$ or $\bm{H} - \bm{E}_{\rm w}$. Since $\bm{1}^T  \bm{c} = 0$, $\Vert (\bm{H} - \bm{E}_{\rm w}) \bm{c} \Vert_{\infty} = \Vert \bm{E}_{\rm w} \bm{c} \Vert_{\infty} \leq r$, which together with $n_{\rm c} \leq n_{\rm u}$, gives $\bm{1}^T (\bm{d}_+^* + \bm{d}_-^*) \leq n_{\rm u} r$.
\end{proof}

For Criterion \ref{crit-0-2-1}, introduce a variable $\phi_1 \in \mathbb{B}$ to indicate if a contingency is a $\lambda$-branch one, which satisfies 
\begin{equation}\label{eq-0-2-12}
    1 - \phi_1 (\lambda + 1)  \leq  (n_{\rm b} - \bm{1}^T \bm{o}_{\rm b}) - \lambda \leq  (1 - \phi_1) (\eta  - \lambda ) 
\end{equation}
such that $\phi_1 = 1$ if $n_{\rm b} - \bm{1}^{T} \bm{o}_{\rm b} \leq \lambda$, namely that the contingency is a $\lambda$-branch one, and $\phi_1 = 0$ otherwise. By Proposition \ref{prop-0-2-1}-\textit{(\uppercase\expandafter{\romannumeral3})}, Criterion \ref{crit-0-2-1} can be formulated as 
\begin{equation}\label{eq-0-2-13}
    \bm{1}^T(\bm{d}_+^* + \bm{d}_-^*) \leq n_{\rm u} r + (1 - \phi_1) M 
\end{equation}

For Criterion \ref{crit-0-2-2}, introduce variables $\phi_2, \phi_3 \in \mathbb{B}$ satisfying
\begin{subequations}\label{eq-0-2-14} 
    \begin{align}
        & M \phi_2 \geq 2 \Vert \bm{b} \Vert_{- 2 \infty}  - \bm{1}^T (\bm{d}_+^* + \bm{d}_-^*) \geq M (\phi_2 - 1)  \\[-2pt]
        & M \phi_3 \geq \bm{1}^T (\bm{d}_+^* + \bm{d}_-^*) - n_{\rm u} r \geq M (\phi_3 - 1) 
    \end{align}
\end{subequations}
such that $\phi_2 = \phi_3 = 0$ iff $2 \Vert \bm{b} \Vert_{- 2 \infty}  \leq \bm{1}^T (\bm{d}_+^* + \bm{d}_-^*) \leq n_{\rm u} r$, i.e., $\tilde{\bm{z}}$ is $\mathcal{W}(\lambda)$ disconnected, by Proposition \ref{prop-0-2-1}. Then Criterion \ref{crit-0-2-2} can be formulated as 
\begin{equation}\label{eq-0-2-15} 
     \bar{\bm{z}} \in \mathcal{C}(\frac{\bm{1}^T (\bm{d}_+^* + \bm{d}_-^*)}{2 \Vert \bm{b} \Vert_{-2 \infty}}, \bm{c}, \bm{0} ), \bar{\bm{z}} \!+\! \bm{1} \!-\! \bm{o}_{\rm b} \!\in\! \mathcal{C}(\phi_2 \!+\! \phi_3, \bm{c}, \bm{0} )   
\end{equation}

Finally, the second-stage corrective control problem in (\ref{eq-0-2-1}) with the above formulation of NC becomes 
\begin{equation}\label{eq-0-2-16} 
    \!\!\!\!\!\! \min_{\bm{x}, \bm{\phi}, \bm{d}_+^*, \bm{d}_-^* } \!\! g(\! \bm{x} \!) ~\text{s.t.}~ \bm{x} \!\in\!\! \mathcal{Z}(\bm{p}_{\rm g},\! \bm{z},\! \bm{o} ), (\ref{eq-0-2-12})\text{-}(\ref{eq-0-2-15}), (\!\bm{d}_+^*, \bm{d}_-^*\!) \!\!\in\!\!\arg (\ref{eq-0-2-5})  \!\!\! 
\end{equation}
 with $\bm{\phi} =[\phi_1, \phi_2, \phi_3]^T \in \mathbb{B}^3$. The bi-level structure of (\ref{eq-0-2-16}) complicates the solution to the resulting SCOTS problem. However, since the lower level of (\ref{eq-0-2-16}) is an LP, it can be replaced by its Karush-Kuhn-Tucker (KKT) conditions, which further with the KKT complementarity condition reformulated in a mixed-integer form, is written as 
 \begin{equation}\label{eq-0-2-17}
         \!\!\! \bm{A} \bm{y} \!\!\leq\!\! \bm{w}(\tilde{\bm{z}}), \bm{A}^{T} \!\bm{\lambda} \!=\! \bm{h}, \bm{\lambda} \!\!\geq\!\! \bm{0}, \bm{w}(\tilde{\bm{z}}) \!-\! \bm{A} \bm{y} \!\!\leq\!\! M (\bm{1} \!-\! \bm{\xi}), \bm{\lambda} \!\!\leq\!\! M \bm{\xi}    
 \end{equation}
where $\bm{y} \!\!=\!\! [{\bm{d}_+^*}^T, {\bm{d}_-^*}^T, \bm{\vartheta}^T, \bm{\rho}^T ]^T$, $\bm{\lambda} \!\!\in\!\! \mathbb{R}^{4(n_{\rm n} \!+ n_{\rm b})}$, $\bm{\xi} \!\!\in\!\! \mathbb{B}^{4(n_{\rm n} + n_{\rm b})}$, and $\bm{h}$, $\bm{w}(\tilde{\bm{z}})$ and $\bm{A}$ are coefficient matrices of the equivalent compact form of (\ref{eq-0-2-5}), i.e., $\min_{\bm{y} \in \mathbb{R}^{n_{\rm b} + 3 n_{\rm n}}} \bm{h}^T \bm{y} ~\text{s.t.}~ \bm{A} \bm{y} \leq \bm{w}(\tilde{\bm{z}})$. Accordingly, (\ref{eq-0-2-16}) is equivalent to 
\begin{equation}\label{eq-0-2-18} 
     \min_{\bm{x}, \bm{\phi}, \bm{y}, \bm{\lambda}, \bm{\xi}}  g( \bm{x} ) ~\text{s.t.}~ \bm{x} \in \mathcal{Z}(\bm{p}_g, \bm{z}, \bm{o} ), (\ref{eq-0-2-12})\text{-}(\ref{eq-0-2-15}), (\ref{eq-0-2-17})   
\end{equation}
Replacing the second-stage problem in (\ref{eq-0-2-1}) by (\ref{eq-0-2-18}) yields the final SCOTS formulation considering Criterion \ref{crit-0-2-1} and Criterion \ref{crit-0-2-2}. Solution approaches for (\ref{eq-0-2-1}) are generally still applicable since the mixed-integer form of the second-stage problem is unchanged after this replacement.

\section{Case Study}

The proposed formulation of NC for SCOTS is tested using 4 systems: IEEE 14-bus, 30-bus, and 57-bus systems, and the 50Hertz control area of the German transmission network (E14, E30, E57, and DE for short). We set $(\eta, \lambda)= (2,1)$ for E14 and E30, and $(\eta, \lambda) = (3, 2)$ for E57 and DE. Model (\ref{eq-0-2-1}) with only NC constraints on $\bm{z}$ and that with the proposed NC formulation are solved respectively. All three versions of SCOTS, i.e., the stochastic, robust and DR, using the DC power flow model, are considered. For stochastic SCOTS, except for E14, only 100 contingency scenarios are considered for computational tractability. For detailed data and parameter setting we refer the reader to \cite{4-971-2}.

\begin{table}[h]
	\centering
    \caption{Results of NC of post-contingency and post-control topology.} 
    \setlength{\tabcolsep}{0pt} 

    \setlength{\aboverulesep}{0pt}
    \setlength{\belowrulesep}{0pt}
    \setlength{\extrarowheight}{.0ex}
 
    \begin{tabular*}{\hsize}{@{}p{0.77cm}@{}P{0.62cm}@{}>{\columncolor{gray!25}}P{0.6cm}@{}P{0.6cm}@{}>{\columncolor{gray!25}}P{0.6cm}@{}P{0.6cm}@{}>{\columncolor{gray!25}}P{0.6cm}@{}P{0.6cm}@{}>{\columncolor{gray!25}}P{0.6cm}@{}P{0.6cm}@{}>{\columncolor{gray!25}}P{0.6cm}@{}P{0.6cm}@{}>{\columncolor{gray!25}}P{0.6cm}}\toprule
    & \multicolumn{4}{c}{Stochastic SCOTS} & \multicolumn{4}{c}{Robust SCOTS} & \multicolumn{4}{c}{DR SCOTS}  
    \\\cmidrule(l){2-5}\cmidrule(l){6-9}\cmidrule(lr){10-13} 
               & $\tilde{r}$ & $\tilde{r}$  & $\bar{r}$ & $\bar{r}$ & $\tilde{r}$ & $\tilde{r}$  & $\bar{r}$ & $\bar{r}$ & $\tilde{r}$ & $\tilde{r}$  & $\bar{r}$ & $\bar{r}$  \\ \midrule[1pt]
    E14  & 25   & 0 & 7.79 & 0 & 30    & 0 & 5.45 & 0 & 25 & 0 & 8.37 & 0 \\
    E30  & 9.76 & 0 & 6.58 & 0 & 17.1  & 0 & 9.29 & 0 & 12.2   & 0  & 8.13 & 0  \\
    E57  & 2.22 & 0 & 11.2 & 0 & 6.49  & 0 & 13.7  & 0    & 8.60 & 0 & 6.69 & 0  \\
    DE   & 4.44 & 0 & 3.41 & 0 & 5.52  & 0 & 9.43  & 0    & 7.38  & 0 & 10.5 & 0  \\ \bottomrule
    \multicolumn{13}{@{}p{1\columnwidth}@{}}{
        \footnotesize{\textit{Note}: $\tilde{r}\%$ denotes the proportion of $\lambda$-branch contingencies whose corresponding $\tilde{\bm{z}}$ is disconnected but not inevitable disconnected; $\bar{r}\%$ is the proportion of contingencies with corresponding $\bar{\bm{z}}$ being connected or inevitable disconnected and $\bm{1}^{T}(\bm{z}_+ + \bm{z}_-) \neq 0$, where corrective line switching disconnects the network further. White and grey rows are results for model (\ref{eq-0-2-1}) with only NC constraints on $\bm{z}$ and that with the proposed NC formulation, respectively. Since it can be time-consuming to obtain the exact value of $\bar{r}$, except for E14, E30, and stochastic SCOTS, $\bar{r}$ is estimated by the Monte Carlo method and the same contingency samples are used to evaluate $\bar{r}$ for the two cases of (\ref{eq-0-2-1}).
      }}
    \end{tabular*} 
    \label{tab-0-2-1}  
\end{table}

Table \ref{tab-0-2-1} gives the statistical results on NC of $\tilde{\bm{z}}$ and $\bar{\bm{z}}$ in the optimum of each model. It is found that for all test systems and versions of SCOTS, $\tilde{r}$ and $\bar{r}$ are both larger than 0 when only NC constants on $\bm{z}$ is considered in the SCOTS model, while when the proposed NC constraints on $\tilde{\bm{z}}$ and $\bar{\bm{z}}$ are incorporated, $\tilde{r}$ and $\bar{r}$ both equal to 0. This indicates that in the former case, due to line switching in $\bm{z}$, some contingencies cause ND, which is not allowed by Criterion \ref{crit-0-2-1}. In fact, an $N-1$ contingency just causes such ND, according to the results on E14 and E30. In addition, corrective line switching also causes new ND when $\tilde{\bm{z}}$ is connected or inevitably disconnected. In contrast, by incorporating the proposed NC constraints into SCOTS models, $\tilde{\bm{z}}$ is always connected ignoring inevitable ND, and corrective line switching does not deteriorate NC when $\tilde{\bm{z}}$ is connected or inevitably disconnected.

\section{Conclusion}

In this paper, we addressed the issue of NC in SCOTS problems by proposing two criteria for NC and developing their tractable mathematical formulation. The criteria are aimed at preserving NC of post-contingency and post-control topologies as much as possible within reasonable limits. The mathematical formulation, in mixed-integer linear form, maintains the applicability of original solution approaches to SCOTS models. Numerical tests show that adding the proposed formulation into SCOTS models guarantees the satisfaction of the criteria which is unachievable with only NC constraints on $\bm{z}$. Future work should focus on existence of trivial special cases of $\mathcal{W}(\lambda)$-uniquely balanced $\bm{c}$ and heuristics to solve (\ref{eq-0-2-c-milp}).


\ifCLASSOPTIONcaptionsoff
  \newpage
\fi

\bibliographystyle{IEEEtran}
\bibliography{References/4}

\end{document}

%% file: Preamble/preamble.tex
\usepackage{booktabs} 
\usepackage{multirow}
\usepackage{mathtools}

\usepackage{amssymb}

\usepackage{nicematrix}

\usepackage{enumitem}

\usepackage{smartdiagram}
\usesmartdiagramlibrary{additions}

\graphicspath{{Figs/}}

\usepackage{pdfpages}
\usepackage{pdflscape}
\usepackage{amsmath}
\usepackage{amsthm}

\usepackage{amsfonts}

\usepackage{color}

\usepackage{algorithm}
\usepackage{algorithmic}
\let\Algorithm\algorithm

\renewcommand\algorithm[1][]{\Algorithm[#1]\setstretch{1}}

\usepackage{setspace}
\usepackage{enumitem}
\usepackage{xcolor}

\usepackage{bm}
\usepackage{bbm}

\usepackage{tikz}
\usepackage{pgf}
\usetikzlibrary{arrows,automata}
\usetikzlibrary{shapes}
\tikzstyle{data44}=[rectangle split,rectangle split parts=2,draw,text centered]

\usepackage{subfigure}

\DeclareMathAlphabet{\mathcal}{OMS}{cmsy}{m}{n}

\usetikzlibrary{matrix}

\tikzset{
  BarreStyle/.style =   {opacity=.3,line width=14 mm,color=#1},
  node style ge/.style={},
  node style sp/.style={},
  yl/.style={},
  arrow style mul/.style={},
}

\newtheorem{proposition}{Proposition}

\newtheorem{remark}{Remark}
\newtheorem{definition}{Definition}

\newtheorem{criterion}{Criterion}

\usepackage{makecell}

\usepackage{threeparttable,booktabs}

\usepackage[colorlinks,
            linkcolor=blue,
            anchorcolor=blue,
            citecolor=blue
            ]{hyperref}

\usepackage{cite}

\usepackage{mdwlist}

\newcolumntype{P}[1]{>{\centering\arraybackslash}p{0.66cm}}

\usepackage{colortbl}